\documentclass[letterpaper, 10 pt,journal]{ieeetran}

\IEEEoverridecommandlockouts                              

\usepackage{graphics} 
\usepackage{graphicx}

\usepackage{times} 
\usepackage{amsmath} 
\usepackage{amssymb}  

\usepackage{tikz} 
\usepackage{pgfplots} 

\usepackage{algpseudocode}%
\usepackage{algorithm}
\usepackage{epstopdf}

\usepackage{cite}

\usepackage{color}

\usepackage{lipsum}
\usepackage{verbatim}

\usepackage{amsthm}

\usepackage{mathtools}

\newtheorem{theorem}{\bf Theorem}[section] \newtheorem{definition}{\bf Definition}[section]
 
 \newtheorem{corollary}{\bf Corollary}[section] \newtheorem{proposition}{\bf Proposition}[section] 
\newtheorem{assumption}{\bf Assumption}[section]

\usepackage{bbm}

\newcommand\rmc{\mathrm{c}}
\newcommand\rmd{\mathrm{d}}

\newcommand\rml{\mathrm{l}}

\newcommand\rmu{\mathrm{u}}

\newcommand\rmx{\mathrm{x}}

\newcommand\rmV{\mathrm{V}}

\newcommand\bbi{\mathbb{I}}

\newcommand\bbr{\mathbb{R}}

\newcommand\bbu{\mathbb{U}}

\newcommand\bby{\mathbb{Y}}

\newcommand\calK{\mathcal{K}}

\title{\LARGE \bf
Stability in data-driven MPC:
an inherent robustness perspective
}

\author{
Julian Berberich$^1$, Johannes K\"ohler$^2$, Matthias A. M\"uller$^3$, Frank Allg\"ower$^1$
\thanks{
F. Allgöwer is thankful that his work was funded by Deutsche Forschungsgemeinschaft (DFG, German Research Foundation) under Germany's Excellence Strategy - EXC 2075 - 390740016 and under grant 468094890. 
		F. Allgöwer acknowledges the support by the Stuttgart Center for Simulation Science (SimTech).
		M. A. Müller is thankful that his work was funded by the European Research Council (ERC) under the European Union’s Horizon 2020 research and innovation programme (grant agreement No 948679).
		J. Berberich thanks the International Max Planck Research School for Intelligent Systems (IMPRS-IS)
for supporting him.}
\thanks{$^1$University of Stuttgart, Institute for Systems Theory and Automatic Control, 70550 Stuttgart, Germany (email:$\{$julian.berberich, frank.allgower$\}$@ist.uni-stuttgart.de)}
\thanks{$^2$Institute for Dynamical Systems and Control, ETH Zurich, ZH-8092, Switzerland (email:jkoehle@ethz.ch)}
\thanks{$^3$Leibniz University Hannover, Institute of Automatic Control, 30167 Hannover, Germany (e-mail:mueller@irt.uni-hannover.de)}}

\begin{document}
\IEEEpubid{\begin{minipage}{\textwidth}\ \\[12pt] \\ \\
\copyright 2022 IEEE. This version has been accepted for publication in Proc. IEEE Conference on Decision and Control (CDC), 2022. Personal use of this material is permitted. Permission from IEEE must be obtained for all other uses, in any current or future media, including reprinting/republishing this material for advertising or promotional purposes, creating new collective works, for resale or redistribution to servers or lists, or reuse of any copyrighted component of this work in other works.
\end{minipage}}

\maketitle

\begin{abstract}
Data-driven model predictive control (DD-MPC) based on Willems' Fundamental Lemma has received much attention in recent years, allowing to control systems directly based on an implicit data-dependent system description.
The literature contains many successful practical applications as well as theoretical results on closed-loop stability and robustness.
In this paper, we provide a tutorial introduction to DD-MPC for unknown linear time-invariant (LTI) systems with focus on (robust) closed-loop stability.
We first address the scenario of noise-free data, for which we present a DD-MPC scheme with terminal equality constraints and derive closed-loop properties.
In case of noisy data, we introduce a simple yet powerful approach to analyze robust stability of DD-MPC by combining continuity of DD-MPC w.r.t.\ noise with inherent robustness of model-based MPC, i.e., robustness of nominal MPC w.r.t.\ small disturbances.
Moreover, we discuss how the presented proof technique allows to show closed-loop stability of a variety of DD-MPC schemes with noisy data, as long as the corresponding model-based MPC is inherently robust.
\end{abstract}


\section{Introduction}
Willems' Fundamental Lemma~\cite{willems2005note} is a foundational result from behavioral systems theory that allows to parametrize all trajectories of a linear time-invariant (LTI) system based on one data trajectory with persistently exciting input component.
This parametrization lends itself naturally to designing controllers based directly on data, see~\cite{markovsky2021behavioral} for an extensive survey.
One prominent application is the design of model predictive control (MPC) schemes~\cite{rawlings2020model}, which can handle general performance criteria and input, state, or output constraints.
Here, the state-space model, which is commonly used to optimize over predicted trajectories, is replaced by the implicit data-driven system representation~\cite{yang2015data,coulson2019deepc}.
This direct data-driven control procedure has potential advantages if compared to the more established, indirect procedure of first identifying a model and then applying model-based MPC.
In particular, data-driven MPC (DD-MPC) is simple to implement in the sense that no intermediate model identification is required and it yields good empirical results in complex nonlinear control applications~\cite{huang2019power,elokda2021quadcopters,berberich2021at}.
Furthermore, DD-MPC admits strong theoretical guarantees in open~\cite{coulson2021distributionally,huang2021robust,yin2021maximum,pan2021stochastic} and closed~\cite{berberich2021guarantees,berberich2020constraints,berberich2021on,berberich2021linearpart2_extended,bongard2022robust,schmitz2022willems,kloeppelt2022novel,alsalti2022data} loop, even for noisy data or nonlinear systems, both scenarios in which the theoretical analysis of identification-based approaches is challenging.

In this paper, we provide a tutorial introduction to DD-MPC based on the Fundamental Lemma~\cite{willems2005note} with a focus on closed-loop stability guarantees for both noise-free and noisy data.
We present a generic framework for the theoretical analysis of robust DD-MPC schemes with noisy data which relies on \emph{inherent robustness}.
A (model-based) nominal MPC scheme is referred to as inherently robust if it is robust w.r.t.\ small disturbances, without resorting to explicit robustifications as in robust MPC~\cite[Section~3.5]{rawlings2020model}.
Various works have studied inherent robustness of nominal MPC with terminal constraints~\cite{yu2014inherent} and without terminal constraints~\cite[Theorem 7.26]{gruene2017nonlinear}, see~\cite{limon2009input,roset2008robustness,messina2005discrete} for more general results.

\begin{figure}
\begin{center}
\begin{tikzpicture}
\node at (0,0) {\textbf{Robust DD-MPC}};
\node at (4,0) {\textbf{Model-based MPC}};
\draw [thick] (-2,-0.5) -- (6,-0.5);
\draw [thick,dashed] (2,0.5) -- (2,-5.5);
\draw[rounded corners=4pt,fill=black!15!white] (-1.7,-0.8) rectangle (1,-2.5);
\node at (-0.35,-1.2) {\textbf{MPC scheme}};
\node at (-0.35,-1.6) {with noisy data};
\node at (-0.35,-2) {$\tilde{y}=y+\varepsilon$};
\draw [thick,->] (1,-1.65)--(3,-1.65) node[midway,above] {\textbf{Thm. IV.1}$\,$};
\draw[rounded corners=4pt,fill=black!15!white] (3,-0.8) rectangle (5.7,-2.5);
\node at (4.35,-1.2) {\textbf{MPC scheme}};
\node at (4.35,-1.6) {with disturbance};
\node at (4.35,-2) {$u=\bar{u}^*+d$};
\draw [thick,->] (4.35,-2.5)--(4.35,-3.8) node[midway,above,sloped] {\textbf{Prop.}};
\draw [thick,->] (4.35,-2.5)--(4.35,-3.8) node[midway,below,sloped] {\textbf{IV.1}};
\draw[rounded corners=4pt,fill=black!15!white] (3,-3.8) rectangle (5.7,-5.5);
\node at (4.35,-4.2) {\textbf{Inherent}};
\node at (4.35,-4.6) {\textbf{robustness}};
\node at (4.35,-5) {w.r.t.\ $d$};
\draw [thick,->] (3,-4.65)--(1,-4.65) node[midway,above] {\textbf{Cor. IV.1}};
\draw[rounded corners=4pt,fill=black!15!white] (-1.7,-3.8) rectangle (1,-5.5);
\node at (-0.35,-4.2) {\textbf{Practical}};
\node at (-0.35,-4.6) {\textbf{stability}};
\node at (-0.35,-5) {w.r.t.\ $\varepsilon$};
\end{tikzpicture}
\end{center}
\caption{Main idea of the robust stability proof of DD-MPC.}
\label{fig:idea}
\end{figure}
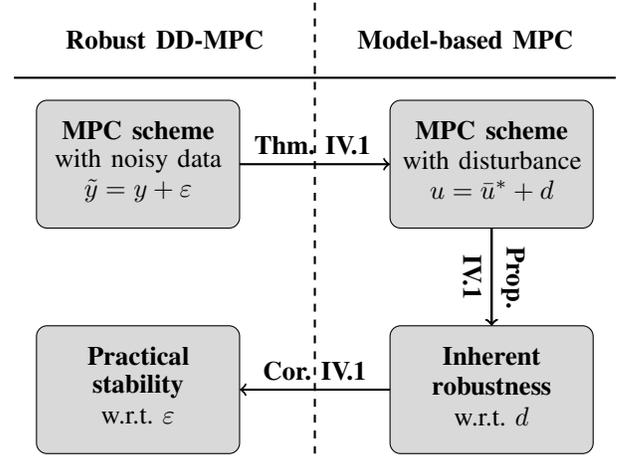

Our theoretical analysis involves a two-step procedure, see Figure~\ref{fig:idea}: 
First, we prove continuity of DD-MPC in the sense that output measurement noise can be translated into an additive input disturbance for the corresponding model-based MPC scheme.
Then, we employ inherent robustness properties of the latter to prove practical stability of the original DD-MPC scheme.
Notably, the robust stability proof relies on the same Lyapunov function used to prove stability of the model-based MPC.
The main advantage of our theoretical analysis is that it directly generalizes to a wide range of DD-MPC schemes, as long as the corresponding model-based MPC scheme is inherently robust.
The proof strategy in this paper originates from and extends~\cite{berberich2021linearpart2_extended}, which shows stability guarantees of a data-driven tracking MPC scheme for nonlinear systems.

\newpage
The paper is structured as follows.
After introducing some preliminaries in Section~\ref{sec:prelim}, we analyze DD-MPC with noise-free and noisy data in Sections~\ref{sec:nominal} and~\ref{sec:robust}, respectively.
Section~\ref{sec:conclusion} concludes the paper.

\subsubsection*{Notation}
We define $\bbi_{\geq0}$ as the set of nonnegative integers and $\bbi_{[a,b]}$ as the set of integers in the interval $[a,b]$.
We write $A=A^\top\succ0$ if $A$ is positive definite.
For a set of symmetric matrices $\{A_i\}_{i=1}^n$, $\lambda_{\min}(A_1,\dots,A_n)$ denotes the smallest of all eigenvalues of the $A_i$'s (and similarly for $\lambda_{\max}$).
For a vector $x\in\bbr^n$ and a matrix $P\succ0$, we define $\lVert x\rVert_P=\sqrt{x^\top Px}$ and $\lVert x\rVert_2=\sqrt{x^\top x}$.
We write $\calK_{\infty}$ for the set of continuous functions $\beta:\bbr_{\geq0}\to\bbr_{\geq0}$ which are strictly increasing, unbounded, and satisfy $\beta(0)=0$.
For a sequence $\{u_k\}_{k=0}^{N-1}$, we define the Hankel matrix
\begin{align*}
H_L(u)\coloneqq\begin{bmatrix}u_0&u_1&\dots&u_{N-L}\\
u_1&u_2&\dots&u_{N-L+1}\\
\vdots&\vdots&\ddots&\vdots\\
u_{L-1}&u_L&\dots&u_{N-1}
\end{bmatrix}
\end{align*}
as well as a window $u_{[a,b]}\coloneqq\begin{bmatrix}u_a^\top&\dots&u_b^\top\end{bmatrix}^\top$.
Further, we write $u=u_{[0,N-1]}$ for the stacked vector containing all entries of the sequence.

\section{Preliminaries}\label{sec:prelim}
We consider an LTI system
\begin{align}\label{eq:sys_LTI}
x_{k+1}&=Ax_k+Bu_k,\\\nonumber
y_k&=Cx_k+Du_k
\end{align}
with state $x_k\in\bbr^n$, input $u_k\in\bbr^m$, and output $y_k\in\bbr^p$, all at time $k\in\bbi_{\geq0}$.
We make the standing assumption that $(A,B)$ is controllable and $(A,C)$ is observable.
Throughout the paper, the matrices $A$, $B$, $C$, $D$, are unknown, but one input-output data trajectory $\{u_k^\rmd,y_k^\rmd\}_{k=0}^{N-1}$ of~\eqref{eq:sys_LTI} is available, which is noise-free (Section~\ref{sec:nominal}) or noisy (Section~\ref{sec:robust}).
The corresponding input component will be assumed to be persistently exciting in the following (standard) sense~\cite{willems2005note}.
\begin{definition}\label{def:pe}
We say that $\{u_k^\rmd\}_{k=0}^{N-1}$ is persistently exciting of order $L$ if $\mathrm{rank}(H_L(u^\rmd))=mL$.
\end{definition}
We now introduce the \emph{Fundamental Lemma} by~\cite{willems2005note}.
\begin{theorem}{\normalfont\cite{willems2005note}}\label{thm:willems}
Suppose $\{u_k^\rmd\}_{k=0}^{N-1}$ is persistently exciting of order $L+n$.
Then, $\{u_k,y_k\}_{k=0}^{L-1}$ is a trajectory of~\eqref{eq:sys_LTI} if and only if there exists $\alpha\in\bbr^{N-L+1}$ such that
\begin{align}\label{eq:thm_willems}
\begin{bmatrix}H_L(u^\rmd)\\H_L(y^\rmd)\end{bmatrix}\alpha=\begin{bmatrix}u\\y\end{bmatrix}.
\end{align}
\end{theorem}
Theorem~\ref{thm:willems} parametrizes all trajectories of the LTI system~\eqref{eq:sys_LTI}, using only measured data and no explicit model knowledge.
The result is at the core of numerous recent data-driven control approaches, see~\cite{markovsky2021behavioral} for an overview, and will be used to set up DD-MPC schemes in the present paper.
Since we only have access to input-output data, we (implicitly) work with the extended, non-minimal state
\begin{align}\label{eq:xi_def}
\xi_t\coloneqq\begin{bmatrix}u_{[t-n,t-1]}^\top&y_{[t-n,t-1]}^\top\end{bmatrix}^\top\in\bbr^{(m+p)n}.
\end{align}
While Theorem~\ref{thm:willems} as well as the definition of $\xi_t$ require knowledge of the system order $n$, the value $n$ can be replaced by an arbitrary upper bound.

\section{Data-driven MPC with noise-free data}\label{sec:nominal}
\begin{subequations}\label{eq:DD_MPC}
Our control goal is stabilization of the origin\footnote{Non-zero setpoints can be considered with straightforward modifications and are omitted for brevity.} of~\eqref{eq:sys_LTI} while satisfying pointwise-in-time constraints on the input and the output, i.e., $u_t\in\bbu$, $y_t\in\bby$ for all $t\in\bbi_{\geq0}$ with closed sets $\bbu\subseteq\bbr^m$, $\bby\subseteq\bbr^p$.
To this end, we design a DD-MPC scheme based on Theorem~\ref{thm:willems}.
At time $t\in\bbi_{\geq0}$ and for given initial conditions $\{u_k,y_k\}_{k=t-n}^{t-1}$, we consider the following optimal control problem
\begin{align}\label{eq:DD_MPC_cost}
\underset{\alpha(t),\bar{u}(t),\bar{y}(t)}{\min}&\sum_{k=0}^{L-1}
\lVert\bar{u}_k(t)\rVert_R^2+\lVert\bar{y}_k(t)\rVert_Q^2\\
\label{eq:DD_MPC_hankel} \text{s.t.}\>\> &\>\begin{bmatrix}
\bar{u}(t)\\\bar{y}(t)\end{bmatrix}=\begin{bmatrix}H_{L+n}(u^\rmd)\\H_{L+n}(y^\rmd)\end{bmatrix}\alpha(t),\\\label{eq:DD_MPC_init}
&\>\begin{bmatrix}\bar{u}_{[-n,-1]}(t)\\\bar{y}_{[-n,-1]}(t)\end{bmatrix}=\begin{bmatrix}u_{[t-n,t-1]}\\y_{[t-n,t-1]}\end{bmatrix},\\\label{eq:DD_MPC_constraints}
&\>\bar{u}_k(t)\in\mathbb{U},\>\bar{y}_k(t)\in\mathbb{Y},\>k\in\mathbb{I}_{[0,L]},
\\\label{eq:DD_MPC_TEC}
&\>\bar{u}_k(t)=0,\>\bar{y}_k(t)=0,\>k\in\bbi_{[L-n,L-1]}.
\end{align}
\end{subequations}
As in standard (model-based) MPC~\cite{rawlings2020model}, Problem~\eqref{eq:DD_MPC} minimizes the deviation from the setpoint $(u,y)=(0,0)$ over the horizon $L\geq n$, weighted by $Q,R\succ0$.
The constraint~\eqref{eq:DD_MPC_hankel} ensures that the input-output trajectory $(\bar{u}(t),\bar{y}(t))$ predicted at time $t$ is indeed a trajectory of~\eqref{eq:sys_LTI}, compare Theorem~\ref{thm:willems}.
This trajectory is of length $L+n$, since the first $n$ components are required to implicitly fix the initial conditions via the past $n$ input-output measurements in~\eqref{eq:DD_MPC_init}.
Finally, Problem~\eqref{eq:DD_MPC} contains input-output constraints~\eqref{eq:DD_MPC_constraints} as well as terminal equality constraints~\eqref{eq:DD_MPC_TEC} for the extended state $\xi_t$ in order to ensure stability.

We write $J_L^*(\xi_t)$ for the optimal cost of Problem~\eqref{eq:DD_MPC}, where $\xi_t$ is the extended state corresponding to $\{u_k,y_k\}_{k=t-n}^{t-1}$, see~\eqref{eq:xi_def}.
Further, the optimal solution of Problem~\eqref{eq:DD_MPC} at time $t$ is denoted by $\bar{u}^*(t)$, $\bar{y}^*(t)$, $\alpha^*(t)$.
On the other hand, closed-loop variables at time $t$ are written as $u_t$, $y_t$, $\xi_t$.
Problem~\eqref{eq:DD_MPC} is used to implement an MPC scheme in a standard fashion~\cite{rawlings2020model}:
At time $t$, we measure $\{u_k,y_k\}_{k=t-n}^{t-1}$, solve Problem~\eqref{eq:DD_MPC}, and apply the first component of the optimal input, i.e., $u_t=\bar{u}_0^*(t)$.

Let us state the main theoretical assumptions to derive closed-loop stability guarantees.
\begin{assumption}\label{ass:pe}
The input $\{u_k^\rmd\}_{k=0}^{N-1}$ generating the data is persistently exciting of order $L+2n$.
\end{assumption}
We assume persistence of excitation of order $L+2n$, although Theorem~\ref{thm:willems} only requires $L+n$, since the trajectory length in Problem~\eqref{eq:DD_MPC} is $L+n$ due to the initial conditions~\eqref{eq:DD_MPC_init}.
\begin{assumption}\label{ass:upper_bound}
There exists $c_\rmu>0$ such that $J_L^*(\xi)\leq c_\rmu\lVert\xi\rVert_2^2$ for any $\xi$ such that Problem~\eqref{eq:DD_MPC} is feasible.
\end{assumption}
Assuming a quadratic upper bound for $J_L^*(\xi)$ is not restrictive and holds, e.g., for polytopic constraints~\cite{bemporad2002explicit}.
By detectability of the state-space realization corresponding to $\xi$, there exists an input-output-to-state stability (IOSS) Lyapunov function $W(\xi)=\lVert \xi\rVert_P^2$ for some $P\succ0$ satisfying
\begin{align}\label{eq:thm_nominal_proof2}
W(\xi_{t+1})-W(\xi_t)\leq-\lVert \xi_t\rVert_2^2+c_{\mathrm{IOSS}}(\lVert u_t\rVert_R^2+\lVert y_t\rVert_Q^2)
\end{align}
for any feasible trajectory and a suitable $c_{\mathrm{IOSS}}>0$~\cite{cai2008input}.
In the following theoretical analysis, we employ the Lyapunov function candidate
\begin{align}\label{eq:LF_nominal}
V(\xi)=J_L^*(\xi)+\frac{1}{c_{\mathrm{IOSS}}} W(\xi).
\end{align}
\begin{theorem}\label{thm:nominal}
Suppose Assumptions~\ref{ass:pe} and~\ref{ass:upper_bound} hold.
If Problem~\eqref{eq:DD_MPC} is feasible at time $t=0$, then 
(i) it is feasible at any $t\in\bbi_{\geq0}$, 
(ii) the closed loop satisfies the constraints, i.e., $u_t\in\bbu$, $y_t\in\bby$ for all $t\in\bbi_{\geq0}$, and 
(iii) the origin $\xi=0$ is exponentially stable for the resulting closed loop.
\end{theorem}
\begin{proof}
We only provide a sketch of the proof and refer to~\cite{berberich2021guarantees} for further details.
Statements (i) and (ii) can be shown using standard MPC arguments, i.e., shifting the previously optimal solution to construct a candidate solution at the next time step~\cite{rawlings2020model}.
Using this candidate, it is easy to show that
\begin{align}\label{eq:thm_nominal_proof1}
J_L^*(\xi_{t+1})-J_L^*(\xi_t)\leq -\lVert u_t\rVert_R^2-\lVert y_t\rVert_Q^2.
\end{align}
From~\eqref{eq:thm_nominal_proof2} and~\eqref{eq:thm_nominal_proof1}, we infer
\begin{align}
V(\xi_{t+1})-V(\xi_t)\leq-\frac{1}{c_{\mathrm{IOSS}}}\lVert \xi_t\rVert_2^2.
\end{align}
Together with the trivial lower bound $V(\xi)\geq\frac{\lambda_{\min}(P)}{c_{\mathrm{IOSS}}}\lVert \xi\rVert_2^2$ and the upper bound $V(\xi)\leq \left(c_\rmu+\frac{\lambda_{\max}(P)}{c_{\mathrm{IOSS}}}\right)\lVert \xi\rVert_2^2$ (Assumption~\ref{ass:upper_bound}), this implies exponential stability via standard Lyapunov methods.
\end{proof}

Theorem~\ref{thm:nominal} shows that the DD-MPC scheme defined via Problem~\eqref{eq:DD_MPC} exponentially stabilizes the closed loop while satisfying the input-output constraints.
Since the Fundamental Lemma (Theorem~\ref{thm:willems}) provides an exact parametrization of input-output trajectories, the presented MPC scheme is equivalent to a model-based one, allowing for analogous steps in the stability proof~\cite{rawlings2020model}.
The main technical challenge is that the cost~\eqref{eq:DD_MPC_cost} only involves the predicted input and output and is, therefore, only positive \emph{semi}-definite in the internal state $x$.
This necessitates the use of detectability properties via an IOSS Lyapunov function, similar to model-based MPC with positive semidefinite stage cost~\cite{grimm2005model}.

\section{Data-driven MPC with noisy data}\label{sec:robust}
In this section, we extend the results of Section~\ref{sec:nominal} to the noisy case.
To be precise, we assume that the data are affected by output measurement noise, i.e., we have access to $\{u_k^\rmd,\tilde{y}_k^\rmd\}_{k=0}^{N-1}$, where $\tilde{y}_k^\rmd=y_k^\rmd+\varepsilon_k^\rmd$ with $\lVert \varepsilon_k^\rmd\rVert_2\leq\bar{\varepsilon}$, $k\in\bbi_{[0,N-1]}$, for some  $\bar{\varepsilon}>0$.
Similarly, the online output measurements used to specify initial conditions (compare~\eqref{eq:DD_MPC_init}) are noisy, i.e., we measure $\tilde{y}_k=y_k+\varepsilon_k$ with $\lVert\varepsilon_k\rVert_2\leq\bar{\varepsilon}$, $k\in\bbi_{\geq0}$.

In Section~\ref{subsec:robust_continuity}, we derive the main technical result translating noise in DD-MPC into an input disturbance for model-based MPC.
In Section~\ref{subsec:robust_stab}, we then combine this result with inherent robustness of model-based MPC to prove closed-loop (practical) stability.

\subsection{Continuity of data-driven optimal control}\label{subsec:robust_continuity}
In comparison to the noise-free setting (Section~\ref{sec:nominal}), we include the following additional assumption.
\begin{assumption}\label{ass:output_constraints}
The set $\bbu$ is a convex, compact polytope and $\bby=\bbr^p$.
\end{assumption}
\begin{subequations}\label{eq:DD_MPC_robust}
Output constraints require a constraint tightening, compare~\cite{berberich2020constraints,kloeppelt2022novel}, and are omitted for simplicity.
At time $t$ and for a given set of (noisy) initial conditions $\{u_k,\tilde{y}_k\}_{k=t-n}^{t-1}$, we consider the optimization problem
\begin{align}\label{eq:DD_MPC_robust_cost}
\underset{\substack{\hat{\alpha}(t),\hat{\sigma}(t)\\\hat{u}(t),\hat{y}(t)}}{\min}&\sum_{k=0}^{L-1}
\lVert\hat{u}_k(t)\rVert_R^2+\lVert\hat{y}_k(t)\rVert_Q^2+\lambda_{\alpha}\bar{\varepsilon}^{\beta_{\alpha}}\lVert\hat{\alpha}(t)\rVert_2^2\\\nonumber
&+\frac{\lambda_{\sigma}}{\bar{\varepsilon}^{\beta_{\sigma}}}\lVert\hat{\sigma}(t)\rVert_2^2\\
\label{eq:DD_MPC_robust_hankel} \text{s.t.}\>\> &\>\begin{bmatrix}
\hat{u}(t)\\\hat{y}(t)+\hat{\sigma}(t)\end{bmatrix}=\begin{bmatrix}H_{L+n}(u^\rmd)\\H_{L+n}(\tilde{y}^\rmd)\end{bmatrix}\hat{\alpha}(t),\\\label{eq:DD_MPC_robust_init}
&\>\begin{bmatrix}\hat{u}_{[-n,-1]}(t)\\\hat{y}_{[-n,-1]}(t)\end{bmatrix}=\begin{bmatrix}u_{[t-n,t-1]}\\\tilde{y}_{[t-n,t-1]}\end{bmatrix},
\\\label{eq:DD_MPC_robust_TEC}
&\>\hat{u}_k(t)=0,\>\hat{y}_k(t)=0,\>k\in\bbi_{[L-n,L-1]},\\\label{eq:DD_MPC_robust_constraints}
&\>\hat{u}_k(t)\in\mathbb{U},\>k\in\mathbb{I}_{[0,L]}.
\end{align}
\end{subequations}
In Problem~\eqref{eq:DD_MPC_robust}, the noise-free data $y^\rmd$ and initial conditions $y_{[t-n,t-1]}$ appearing in Problem~\eqref{eq:DD_MPC} have been replaced by their noisy counterparts.
To robustify against the noise, Problem~\eqref{eq:DD_MPC_robust} includes a slack variable $\hat{\sigma}(t)$ which relaxes the constraint~\eqref{eq:DD_MPC_robust_hankel} due to the noisy output measurements.
The slack variable is regularized with parameters $\lambda_{\sigma},\beta_{\sigma}>0$ to avoid a large prediction error.
Additionally, $\hat{\alpha}(t)$ is regularized with parameters $\lambda_{\alpha},\beta_{\alpha}>0$ in order to reduce the influence of the noise on the prediction in~\eqref{eq:DD_MPC_hankel}.
We note that similar modifications of DD-MPC to cope with noise were suggested by~\cite{coulson2019deepc,coulson2021distributionally,huang2021robust,berberich2021guarantees,bongard2022robust}.
In case of noise-free data, i.e., for $\bar{\varepsilon}\to0$, Problem~\eqref{eq:DD_MPC_robust} reduces to the nominal one~\eqref{eq:DD_MPC}.
Throughout this section, the optimization variables associated to Problem~\eqref{eq:DD_MPC_robust} are written as $\hat{\alpha}(t)$, $\hat{\sigma}(t)$, $\hat{u}(t)$, $\hat{y}(t)$, whereas we denote the optimization variables corresponding to Problem~\eqref{eq:DD_MPC} by $\alpha(t)$, $\bar{u}(t)$, $\bar{y}(t)$.
In particular, we write $\hat{\alpha}^*(t)$, $\hat{\sigma}^*(t)$, $\hat{u}^*(t)$, $\hat{y}^*(t)$ for the optimal solution of Problem~\eqref{eq:DD_MPC_robust} at time $t$.
Further, we denote the corresponding optimal cost by $\hat{J}_L^*(\tilde{\xi}_t)$, where $\tilde{\xi}_t\coloneqq\begin{bmatrix}u_{[t-n,t-1]}^\top&\tilde{y}_{[t-n,t-1]}^\top\end{bmatrix}^\top$ is the noisy extended state, compare~\eqref{eq:xi_def}.

\begin{assumption}\label{ass:continuity}
Problem~\eqref{eq:DD_MPC} satisfies a linear independence constraint qualification (LICQ), i.e., the row entries of the equality and active inequality constraints are linearly independent.
Moreover, $\beta_{\alpha}+\beta_{\sigma}<2$.
\end{assumption}

Assuming an LICQ is common in linear-quadratic MPC~\cite{bemporad2002explicit}, and relaxing this assumption is an interesting direction for future research.
The condition $\beta_{\alpha}+\beta_{\sigma}<2$ is required for a technical argument in the proof and can be satisfied by design.
The following result shows that the output measurement noise in Problem~\eqref{eq:DD_MPC_robust} translates into an input disturbance for the nominal problem~\eqref{eq:DD_MPC}.
The statement as well as the subsequent stability analysis rely on the Lyapunov function $V(\xi)$ defined in~\eqref{eq:LF_nominal}, which is used to prove stability in the nominal (equivalently, model-based) case, compare Theorem~\ref{thm:nominal}.
\begin{theorem}\label{thm:continuity}
Suppose Assumptions~\ref{ass:pe},~\ref{ass:output_constraints}, and~\ref{ass:continuity} hold.
Then, for any $\overline{V}>0$, there exists $\beta_\rmu\in\calK_{\infty}$ such that, if $V(\xi_t)\leq\overline{V}$ then
\begin{align}\label{eq:thm_continuity}
\lVert\hat{u}^*(t)-\bar{u}^*(t)\rVert_2\leq\beta_\rmu(\bar{\varepsilon}).
\end{align}
\end{theorem}
The proof can be found in the appendix and relies on three steps:
i) Bounding the optimal cost $\hat{J}_L^*(\tilde{\xi}_t)$ of the robust problem~\eqref{eq:DD_MPC_robust} in terms of the optimal cost $J_L^*(\xi_t)$ of the nominal problem~\eqref{eq:DD_MPC};
ii) relating the minimizer of the robust problem~\eqref{eq:DD_MPC_robust} to the minimizer of a perturbed version of the nominal problem~\eqref{eq:DD_MPC}, where the perturbation vanishes for $\bar{\varepsilon}\to0$;
and iii) relating the minimizer of the latter perturbed problem to the unperturbed nominal problem~\eqref{eq:DD_MPC} via sensitivity analysis of quadratic programs.

Theorem~\ref{thm:continuity} shows that the difference between the optimal inputs $\bar{u}^*(t)$ and $\hat{u}^*(t)$ generated by Problem~\eqref{eq:DD_MPC} and Problem~\eqref{eq:DD_MPC_robust}, respectively, is bounded by the noise level.
The result extends~\cite[Proposition 1]{berberich2021linearpart2_extended} with the technical difference that the considered nominal MPC problem~\eqref{eq:DD_MPC} contains no regularization of $\alpha$ in the cost and is, hence, equivalent to a standard model-based MPC.
According to Theorem~\ref{thm:continuity}, measurement noise in DD-MPC is equivalent to an input disturbance for model-based MPC.
Thus, if we can show that Problem~\eqref{eq:DD_MPC} is inherently robust w.r.t.\ input disturbances, then we can conclude robustness of Problem~\eqref{eq:DD_MPC_robust} w.r.t.\ noisy data.
Following this idea, we prove (practical) stability of robust DD-MPC in the next section.

\subsection{Closed-loop stability and robustness}\label{subsec:robust_stab}
Due to the terminal equality constraint~\eqref{eq:DD_MPC_TEC}, a nominal MPC scheme based on Problem~\eqref{eq:DD_MPC} is in general not inherently robust (or only locally).
Therefore, we consider Problem~\eqref{eq:DD_MPC_robust} in a \emph{multi-step} implementation, compare~\cite{gruene2015robustness,worthmann2017interaction}:
At time $t=n i$, $i\in\bbi_{\geq0}$, we measure $\tilde{\xi}_t$, solve Problem~\eqref{eq:DD_MPC_robust}, and apply $u_{[t,t+n-1]}=\hat{u}^*_{[0,n-1]}(t)$ over the next $n$ time steps.

We begin by showing that the nominal data-driven (equivalently, model-based) MPC with terminal equality constraints is inherently robust when applied in a multi-step ($n$-step) fashion.
To this end, we require the following (mild) assumption.
\begin{assumption}\label{ass:robust_2}
We have $0\in\mathrm{int}(\bbu)$ and $L\geq2n$.
\end{assumption}
\begin{proposition}\label{prop:nominal_robustness}
Suppose Assumptions~\ref{ass:pe},~\ref{ass:upper_bound},~\ref{ass:output_constraints}, and~\ref{ass:robust_2} hold.
Consider System~\eqref{eq:sys_LTI} controlled by an $n$-step MPC scheme based on Problem~\eqref{eq:DD_MPC}, where the input applied to~\eqref{eq:sys_LTI} is perturbed as
\begin{align}
u_{[t,t+n-1]}=\bar{u}^*_{[0,n-1]}(t)+d_{[t,t+n-1]}
\end{align}
for $t=ni$, $i\in\bbi_{\geq0}$.

Then, there exists $\bar{d}>0$ such that, for any disturbance $\{d_t\}_{t=0}^{\infty}$ satisfying $\sup_{t\in\bbi_{\geq0}}\lVert d_t\rVert_2\leq\bar{d}$, if Problem~\eqref{eq:DD_MPC} is feasible at initial time $t=0$, then it is feasible at any time $t=ni$, $i\in\bbi_{\geq0}$.

Furthermore, for any $\overline{V}>0$, there exist $\bar{d}_{\max},\bar{c}_\rml,\bar{c}_\rmu>0$, $0<c_\rmV<1$, and $\beta_\rmd\in\calK_{\infty}$ such that for all initial conditions with $V(\xi_0)\leq \overline{V}$, all $\bar{d}\leq\bar{d}_{\max}$, and all times $t=ni$, $i\in\bbi_{\geq0}$, the closed loop satisfies
\begin{align}\label{eq:prop_nominal_robustness_lower_upper_bound}
\bar{c}_\rml\lVert \xi_t\rVert_2^2&\leq V(\xi_t)\leq \bar{c}_\rmu\lVert \xi_t\rVert_2^2,\\
\label{eq:prop_nominal_robustness_decay}
V(\xi_{t+n})&\leq c_\rmV V(\xi_t)+\beta_\rmd(\bar{d}).
\end{align}
\end{proposition}
\begin{proof}
\textbf{(i). Recursive feasibility:}
We prove recursive feasibility via a candidate solution that results from shifting the previously optimal solution and appending a deadbeat controller to account for the input disturbance.
At time $t+n$, we define the input candidate as $\bar{u}_k'(t+n)=\bar{u}_{k+n}^*(t)$ for $k\in\bbi_{[0,L-2n-1]}$.
The initial conditions~\eqref{eq:DD_MPC_init} imply $\bar{u}_k'(t+n)=u_{t+n+k}$, $\bar{y}_k'(t+n)=y_{t+n+k}$ for $\bbi_{[-n,-1]}$.
Thus, it only remains to define $\bar{u}_k'(t+n)$ for $k\in\bbi_{[L-2n,L-1]}$ as well as the output candidate.

We write $\{\check{y}_k(t+n)\}_{k=0}^{L-n-1}$ for the output trajectory that results from applying the input $\{\bar{u}_{k+n}^*(t)\}_{k=0}^{L-n-1}$ to System~\eqref{eq:sys_LTI} with initial state $x_{t+n}$.
For time steps $k\in\bbi_{[0,L-2n-1]}$, we choose the output candidate as $\bar{y}_k'(t+n)=\check{y}_k(t+n)$.
The only difference between $\check{y}_k(t+n)$ and $\bar{y}_{k+n}^*(t)$ is due to the disturbance and thus, by the linear (hence, Lipschitz continuous) system dynamics~\eqref{eq:sys_LTI}, there exists $c_{\rmd,1}>0$ such that
\begin{align}\label{eq:prop_nominal_robustness_proof1}
\lVert\check{y}_k(t+n)-\bar{y}_{k+n}^*(t)\rVert_2\leq c_{\rmd,1}\bar{d}
\end{align}
for $k\in\bbi_{[0,L-n-1]}$.
The terminal equality constraints~\eqref{eq:DD_MPC_TEC} imply $\bar{y}_{k+n}^*(t)=0$ for $k\in\bbi_{[L-2n,L-n-1]}$ and, therefore, using~\eqref{eq:prop_nominal_robustness_proof1}, we obtain
\begin{align*}
\lVert\check{y}_k(t+n)\rVert_2\leq c_{\rmd,1}\bar{d}\quad\text{for}\>\>k\in\bbi_{[L-2n,L-n-1]}.
\end{align*}
Using additionally $\bar{u}_{k+n}^*(t)=0$ for $k\in\bbi_{[L-2n,L-n-1]}$ as well as the linear system dynamics~\eqref{eq:sys_LTI}, the norm of the internal state corresponding to the trajectory $(\bar{u}'(t+n),\check{y}(t+n))$ at time $L-2n$ is bounded by $c_{\rmd,2}\bar{d}$ for some $c_{\rmd,2}>0$.
We now define the input and output candidate over time steps $k\in\bbi_{[L-2n,L-1]}$ using a deadbeat control argument:
By controllability, there exists an input-output trajectory $\{\bar{u}_k'(t+n),\bar{y}_k'(t+n)\}_{k=L-2n}^{L-n-1}$ steering the system to $(\bar{u}_k'(t+n),\bar{y}_k'(t+n))=(0,0)$ for $k=L-n,\dots,L-1$ while satisfying
\begin{align*}
\sum_{k=L-2n}^{L-n-1}\lVert\bar{u}_k'(t+n)\rVert_2^2+\lVert\bar{y}_k'(t+n)\rVert_2^2\leq c_{\rmd,3}\bar{d}
\end{align*}
for some $c_{\rmd,3}>0$.
If $\bar{d}>0$ is sufficiently small, the input satisfies the constraints $\bar{u}_k'(t+n)\in\bbu$ for $k\in\bbi_{[L-2n,L-n-1]}$ due to $0\in\mathrm{int}(\bbu)$.
Moreover, the constructed candidate satisfies the terminal equality constraints~\eqref{eq:DD_MPC_TEC}.
Finally, a corresponding candidate for $\alpha'(t+n)$ satisfying~\eqref{eq:DD_MPC_hankel} exists by Theorem~\ref{thm:willems}.\\
\textbf{(ii). Practical stability:}
The lower and upper bounds in~\eqref{eq:prop_nominal_robustness_lower_upper_bound} are analogous to Theorem~\ref{thm:nominal}.
Using the above candidate solution, it is straightforward to show that
\begin{align*}
J_L^*(\xi_{t+n})-J_L^*(\xi_t)\leq-\sum_{k=0}^{n-1}\left(
\lVert u_{t+k}\rVert_R^2+\lVert y_{t+k}\rVert_Q^2\right)
+\beta_{\rmd}(\bar{d})
\end{align*}
for a linear function $\beta_{\rmd}\in\calK_{\infty}$.
Applying the IOSS property~\eqref{eq:thm_nominal_proof2} as well as the upper bound in~\eqref{eq:prop_nominal_robustness_lower_upper_bound}, we arrive at the decay bound~\eqref{eq:prop_nominal_robustness_decay}.
Using $V(\xi_t)\leq\overline{V}$, this implies
\begin{align}
V(\xi_{t+n})\leq c_\rmV \overline{V}+\beta_\rmd(\bar{d})\leq\overline{V},
\end{align}
where the last inequality holds if $\bar{d}_{\max}$ is sufficiently small.
Hence, Inequality~\eqref{eq:prop_nominal_robustness_decay} holds recursively for all $t=ni$, $i\in\bbi_{\geq0}$, which concludes the proof.
\end{proof}

Proposition~\ref{prop:nominal_robustness} shows that model-based multi-step MPC with terminal equality constraints is robust w.r.t.\ (sufficiently small) input disturbances.
Specifically, the Lyapunov function $V$ used in the nominal stability proof (Theorem~\ref{thm:nominal}) is a \emph{practical Lyapunov function}~\cite{gruene2014asymptotic} in the presence of disturbances.

We now consider the robust DD-MPC scheme based on Problem~\eqref{eq:DD_MPC_robust} which is applied in a multi-step fashion as described above.
The following result combines Theorem~\ref{thm:continuity} and Proposition~\ref{prop:nominal_robustness} to prove practical stability of the resulting closed loop in the presence of noisy output measurements.

\begin{corollary}\label{cor:combination}
Suppose Assumptions~\ref{ass:pe},~\ref{ass:upper_bound}, and~\ref{ass:output_constraints}--\ref{ass:robust_2} hold.
Consider System~\eqref{eq:sys_LTI} controlled by an $n$-step MPC scheme based on Problem~\eqref{eq:DD_MPC_robust}.

Then, for any $\overline{V}>0$, there exist $\bar{\varepsilon}_{\max}$, $c_\rmV>0$, and $\beta_\rmV\in\calK_{\infty}$ such that, for all initial conditions with $V(\xi_0)\leq \overline{V}$, all $\bar{\varepsilon}\leq\bar{\varepsilon}_{\max}$, and all times $t=ni$, $i\in\bbi_{\geq0}$, the closed loop satisfies
\begin{align}\label{eq:cor_combination}
V(\xi_{t+n})\leq c_\rmV V(\xi_t)+\beta_\rmV(\bar{\varepsilon}).
\end{align}
\end{corollary}
\begin{proof}
Given $\overline{V}>0$, consider $\beta_{\rmd}$ and $\bar{d}_{\max}$ from Proposition~\ref{prop:nominal_robustness}.
Choose $\bar{\varepsilon}_{\max}>0$ sufficiently small such that $\beta_{\rmu}(\bar{\varepsilon}_{\max})\leq\bar{d}_{\max}$ with $\beta_{\rmu}$ as in Theorem~\ref{thm:continuity}.
Combining~\eqref{eq:thm_continuity} and~\eqref{eq:prop_nominal_robustness_decay}, we obtain~\eqref{eq:cor_combination} with $\beta_\rmV\coloneqq\beta_\rmd\circ\beta_\rmu$.
As in the proof of Proposition~\ref{prop:nominal_robustness}, we infer $V(\xi_{t+n})\leq \overline{V}$ such that this argument can be applied recursively and~\eqref{eq:cor_combination} holds for all $t=ni$, $i\in\bbi_{\geq0}$.
\end{proof}

Corollary~\ref{cor:combination} proves closed-loop practical exponential stability under the data-driven multi-step MPC scheme based on Problem~\eqref{eq:DD_MPC_robust}.
To be precise,~\eqref{eq:prop_nominal_robustness_lower_upper_bound} together with~\eqref{eq:cor_combination} implies
\begin{align}
\lVert\xi_t\rVert_2^2\leq\frac{\bar{c}_{\rmu}}{\bar{c}_{\rml}}c_{\rmV}^i\lVert\xi_0\rVert_2^2+\frac{1}{\bar{c}_{\rml}}\sum_{j=0}^{i-1}c_{\rmV}^j\beta_\rmV(\bar{\varepsilon})
\end{align}
for any $t=n i$, $i\in\bbi_{\geq0}$.
Using $c_\rmV<1$, this implies that $\xi_t$ converges exponentially to the neighborhood
\begin{align*}
\Xi_{\bar{\varepsilon}}\coloneqq\Big\{\xi\mid\lVert\xi\rVert_2^2\leq\frac{1}{\bar{c}_{\rml}(1-c_{\rmV})}\beta_{\rmV}(\bar{\varepsilon})\Big\}
\end{align*}
of the origin, the size of which increases with the noise level.
Similarly, the size of the region of attraction $\overline{V}$ depends on the noise level $\bar{\varepsilon}$, i.e., for larger values of $\overline{V}$ a smaller value for $\bar{\varepsilon}$ needs to be selected in Corollary~\ref{cor:combination}, compare the proof of Proposition~\ref{prop:nominal_robustness}.
It is also possible to prove that the neighborhood $\Xi_{\bar{\varepsilon}}$ shrinks and $\overline{V}$ can be chosen larger if the minimum singular value of the input data matrix $H_{L+n}(u^\rmd)$ increases, compare~\cite{berberich2021guarantees}.
Hence, the closed-loop guarantees on practical stability of DD-MPC that are proven in Corollary~\ref{cor:combination} depend directly on the data quality.
Analogous to classical inherent robustness results~\cite{yu2014inherent,gruene2017nonlinear,limon2009input,roset2008robustness,messina2005discrete}, our results  only yield \emph{qualitative} guarantees.
In particular, the constants in Corollary~\ref{cor:combination} and the set $\Xi_{\bar{\varepsilon}}$ cannot be easily computed without detailed model knowledge.
Inferring these quantities from data only is an interesting issue for future research.
Note that Corollary~\ref{cor:combination} employs the Lyapunov function $V$ corresponding to the \emph{nominal} data-driven (i.e., model-based) MPC scheme from Section~\ref{sec:nominal} in order to analyze the closed loop of the \emph{robust} DD-MPC scheme.

The presented theoretical analysis provides a generic proof strategy which can be directly transferred to derive robustness guarantees of various DD-MPC schemes, the main requirement being that the corresponding model-based MPC scheme is inherently robust.
More precisely, the continuity property stated in Theorem~\ref{thm:continuity} implies that noisy data can be viewed as an input disturbance for model-based MPC.
The result remains true for different DD-MPC formulations as long as the underlying optimization problem satisfies an LICQ (Assumption~\ref{ass:continuity}), is strongly convex in the input, and the constraints are polytopic.
Furthermore, Proposition~\ref{prop:nominal_robustness} shows that model-based multi-step MPC with terminal equality constraints is inherently robust w.r.t.\ input disturbances.
Analogous results hold for multi-step tracking MPC with terminal equality constraints and an artifical setpoint~\cite[Appendix D]{berberich2021linearpart2_extended} as well as for standard (one-step) implementations of MPC without terminal constraints~\cite[Theorem 7.26]{gruene2017nonlinear}, MPC with general terminal constraints and terminal cost~\cite{yu2014inherent}, and under more general assumptions~\cite{limon2009input,roset2008robustness,messina2005discrete}.
Thus, our two-step analysis implies practical stability of any robust DD-MPC scheme for which Theorem~\ref{thm:continuity} and Proposition~\ref{prop:nominal_robustness} apply.

Stability and robustness of DD-MPC has also been shown recently in~\cite{berberich2021guarantees} and~\cite{bongard2022robust} for MPC schemes with terminal equality constraints and without any terminal constraints, respectively.
Deriving guarantees in the absence of terminal constraints is particularly relevant since most existing applications of DD-MPC~\cite{huang2019power,elokda2021quadcopters} as well as many implementations of model-based MPC~\cite{mayne2013apologia} omit terminal constraints.
The theoretical analysis in the present paper has significant advantages over the ones from~\cite{berberich2021guarantees} and~\cite{bongard2022robust}:
The framework allows for seamless extensions into multiple directions, whereas the analysis in~\cite{berberich2021guarantees} and~\cite{bongard2022robust} is tailored to the specific problem setup and DD-MPC formulation.
Moreover, the overall analysis is substantially shorter.
On the other hand, the tailored approaches from~\cite{berberich2021guarantees,bongard2022robust} yield more insightful and interpretable bounds related to system properties (e.g., controllability and observability), which can even be used to construct a constraint tightening guaranteeing robust output constraint satisfaction~\cite{berberich2020constraints,kloeppelt2022novel}.

To summarize, the presented inherent robustness perspective provides a unifying framework for the robust stability analysis of DD-MPC in the presence of noisy data.

\section{Conclusion}\label{sec:conclusion}
We provided a tutorial introduction to stability and robustness of DD-MPC using an implicit prediction model based on the Fundamental Lemma.
In case of noise-free data, we presented a stability proof based on a detectability condition to address the positive semidefinite cost function.
In the presence of output measurement noise, we then proved that a modified robust DD-MPC scheme is practically stable w.r.t.\ the noise level.
Our analysis consists of two steps:
1) translating noisy data in DD-MPC into an input disturbance for model-based MPC and 2) proving inherent robustness of the latter.
The presented exposition directly applies to a wide class of DD-MPC schemes and, therefore, simplifies the transfer of model-based MPC results to the recent field of DD-MPC.

\bibliographystyle{IEEEtran}   
\bibliography{Literature}

\begin{thebibliography}{10}
\providecommand{\url}[1]{#1}
\csname url@rmstyle\endcsname
\providecommand{\newblock}{\relax}
\providecommand{\bibinfo}[2]{#2}
\providecommand\BIBentrySTDinterwordspacing{\spaceskip=0pt\relax}
\providecommand\BIBentryALTinterwordstretchfactor{4}
\providecommand\BIBentryALTinterwordspacing{\spaceskip=\fontdimen2\font plus
\BIBentryALTinterwordstretchfactor\fontdimen3\font minus
  \fontdimen4\font\relax}
\providecommand\BIBforeignlanguage[2]{{%
\expandafter\ifx\csname l@#1\endcsname\relax
\typeout{** WARNING: IEEEtran.bst: No hyphenation pattern has been}%
\typeout{** loaded for the language `#1'. Using the pattern for}%
\typeout{** the default language instead.}%
\else
\language=\csname l@#1\endcsname
\fi
#2}}

\bibitem{willems2005note}
J.~C. Willems, P.~Rapisarda, I.~Markovsky, and B.~{De Moor}, ``A note on
  persistency of excitation,'' \emph{Syst. Contr. Lett.}, vol.~54, pp.
  325--329, 2005.

\bibitem{markovsky2021behavioral}
I.~Markovsky and F.~D{\"o}rfler, ``Behavioral systems theory in data-driven
  analysis, signal processing, and control,'' \emph{Annual Reviews in Control},
  vol.~52, pp. 42--64, 2021.

\bibitem{rawlings2020model}
J.~B. Rawlings, D.~Q. Mayne, and M.~M. Diehl, \emph{Model Predictive Control:
  Theory, Computation, and Design}.\hskip 1em plus 0.5em minus 0.4em\relax Nob
  Hill Pub, 2020, 3rd printing.

\bibitem{yang2015data}
H.~Yang and S.~Li, ``A data-driven predictive controller design based on
  reduced hankel matrix,'' in \emph{Proc. Asian Control Conference}, 2015, pp.
  1--7.

\bibitem{coulson2019deepc}
J.~Coulson, J.~Lygeros, and F.~D{\"o}rfler, ``Data-enabled predictive control:
  in the shallows of the {DeePC},'' in \emph{Proc. European Control Conf.
  (ECC)}, 2019, pp. 307--312.

\bibitem{huang2019power}
L.~Huang, J.~Coulson, J.~Lygeros, and F.~D{\"o}rfler, ``Data-enabled predictive
  control for grid-connected power converters,'' in \emph{Proc. 58th IEEE Conf.
  Decision and Control (CDC)}, 2019, pp. 8130--8135.

\bibitem{elokda2021quadcopters}
E.~Elokda, J.~Coulson, J.~Lygeros, and F.~D{\"o}rfler, ``Data-enabled
  predictive control for quadcopters,'' \emph{Int. J. Robust and Nonlinear
  Control}, vol.~31, no.~18, pp. 8916--8936, 2021.

\bibitem{berberich2021at}
J.~Berberich, J.~K{\"o}hler, M.~A. M{\"u}ller, and F.~Allg{\"o}wer,
  ``Data-driven model predictive control: closed-loop guarantees and
  experimental results,'' \emph{at-Automatisierungstechnik}, vol.~69, no.~7,
  pp. 608--618, 2021.

\bibitem{coulson2021distributionally}
J.~Coulson, J.~Lygeros, and F.~D{\"o}rfler, ``Distributionally robust chance
  constrained data-enabled predictive control,'' \emph{IEEE Trans. Automat.
  Control}, 2021, doi: 10.1109/TAC.2021.3097706.

\bibitem{huang2021robust}
L.~Huang, J.~Zhen, J.~Lygeros, and F.~D{\"o}rfler, ``Robust data-enabled
  predictive control: tractable formulations and performance guarantees,''
  \emph{arXiv:2105.07199}, 2021.

\bibitem{yin2021maximum}
M.~Yin, A.~Iannelli, and R.~S. Smith, ``Maximum likelihood estimation in
  data-driven modeling and control,'' \emph{IEEE Trans. Automat. Control},
  2021, doi: 10.1109/TAC.2021.3137788.

\bibitem{pan2021stochastic}
G.~Pan, R.~Ou, and T.~Faulwasser, ``On a stochastic fundamental lemma and its
  use for data-driven {MPC},'' \emph{arXiv:2111.13636}, 2021.

\bibitem{berberich2021guarantees}
J.~Berberich, J.~K{\"o}hler, M.~A. M{\"u}ller, and F.~Allg{\"o}wer,
  ``Data-driven model predictive control with stability and robustness
  guarantees,'' \emph{IEEE Trans. Automat. Control}, vol.~66, no.~4, pp.
  1702--1717, 2021.

\bibitem{berberich2020constraints}
------, ``Robust constraint satisfaction in data-driven {MPC},'' in \emph{Proc.
  59th IEEE Conf. Decision and Control (CDC)}, 2020, pp. 1260--1267.

\bibitem{berberich2021on}
------, ``On the design of terminal ingredients for data-driven {MPC},''
  \emph{IFAC-PapersOnLine}, vol.~54, no.~6, pp. 257--263, 2021.

\bibitem{berberich2021linearpart2_extended}
------, ``Linear tracking {MPC} for nonlinear systems part {II}: the
  data-driven case,'' \emph{IEEE Trans. Automat. Control}, 2022, doi:
  10.1109/TAC.2022.3166851, extended version on arXiv.

\bibitem{bongard2022robust}
J.~Bongard, J.~Berberich, J.~K{\"o}hler, and F.~Allg{\"o}wer, ``Robust
  stability analysis of a simple data-driven model predictive control
  approach,'' \emph{IEEE Trans. Automat. Control}, 2022, doi:
  10.1109/TAC.2022.3163110.

\bibitem{schmitz2022willems}
P.~Schmitz, T.~Faulwasser, and K.~Worthmann, ``Willems' fundamental lemma for
  linear descriptor systems and its use for data-driven output-feedback
  {MPC},'' \emph{IEEE Control Systems Lett.}, vol.~6, pp. 2443--2448, 2022.

\bibitem{kloeppelt2022novel}
C.~Kl{\"o}ppelt, J.~Berberich, F.~Allg{\"o}wer, and M.~A. M{\"u}ller, ``A novel
  constraint tightening approach for robust data-driven predictive control,''
  \emph{arXiv:2203.07055}, 2022.

\bibitem{alsalti2022data}
M.~Alsalti, V.~G. Lopez, J.~Berberich, F.~Allg{\"o}wer, and M.~A. M{\"u}ller,
  ``Data-based control of feedback linearizable systems,''
  \emph{arXiv:2204.01148}, 2022.

\bibitem{yu2014inherent}
S.~Yu, M.~Reble, H.~Chen, and F.~Allg{\"o}wer, ``Inherent robustness properties
  of quasi-infinite horizon nonlinear model predictive control,''
  \emph{Automatica}, vol.~50, no.~9, pp. 2269--2280, 2014.

\bibitem{gruene2017nonlinear}
L.~Gr{\"u}ne and J.~Pannek, \emph{Nonlinear Model Predictive Control}.\hskip
  1em plus 0.5em minus 0.4em\relax Springer, 2017.

\bibitem{limon2009input}
D.~Lim{\'o}n, T.~Alamo, D.~M. {de la Pena}, J.~M. Bravo, A.~Ferramosca, and
  E.~F. Camacho, ``Input-to-state stability: a unifying framework for robust
  model predictive control,'' in \emph{Nonlinear Model Predictive Control:
  Towards New Challenging Applications}.\hskip 1em plus 0.5em minus 0.4em\relax
  Springer, 2009, vol. 384, pp. 1--26.

\bibitem{roset2008robustness}
B.~J.~P. Roset, W.~P. M.~H. Heemels, M.~Lazar, and H.~Nijmeijer, ``On
  robustness of constrained discrete-time systems to state measurement
  errors,'' \emph{Automatica}, vol.~44, no.~4, pp. 1161--1165, 2008.

\bibitem{messina2005discrete}
M.~J. Messina, S.~E. Tuna, and A.~R. Teel, ``Discrete-time certainty
  equivalence output feedback: allowing discontinuous control laws including
  those from model predictive control,'' \emph{Automatica}, vol.~41, no.~4, pp.
  617--628, 2005.

\bibitem{bemporad2002explicit}
A.~Bemporad, M.~Morari, V.~Dua, and E.~N. Pistikopoulos, ``The explicit linear
  quadratic regulator for constrained systems,'' \emph{Automatica}, vol.~38,
  no.~1, pp. 3--20, 2002.

\bibitem{cai2008input}
C.~Cai and A.~R. Teel, ``Input--output-to-state stability for discrete-time
  systems,'' \emph{Automatica}, vol.~44, no.~2, pp. 326--336, 2008.

\bibitem{grimm2005model}
G.~Grimm, M.~J. Messina, S.~E. Tuna, and A.~R. Teel, ``Model predictive
  control: for want of a local control {L}yapunov function, all is not lost,''
  \emph{IEEE Trans. Automat. Control}, vol.~50, no.~5, pp. 546--558, 2005.

\bibitem{gruene2015robustness}
L.~Gr{\"u}ne and V.~G. Palma, ``Robustness of performance and stability for
  multistep and updated multistep {MPC} schemes,'' \emph{Discrete and
  Continuous Dynamical Systems}, vol.~35, no.~9, pp. 4385--4414, 2015.

\bibitem{worthmann2017interaction}
K.~Worthmann, M.~W. Mehrez, G.~K.~I. Mann, R.~G. Gosine, and J.~Pannek,
  ``Interaction of open and closed loop control in {MPC},'' \emph{Automatica},
  vol.~82, pp. 243--250, 2017.

\bibitem{gruene2014asymptotic}
L.~Gr{\"u}ne and M.~Stieler, ``Asymptotic stability and transient optimality of
  economic {MPC} without terminal conditions,'' \emph{J. Proc. Contr.},
  vol.~24, pp. 1187--1196, 2014.

\bibitem{mayne2013apologia}
D.~Q. Mayne, ``An apologia for stabilising terminal conditions in model
  predictive control,'' \emph{Int. J. Control}, vol.~86, no.~11, pp.
  2090--2095, 2013.

\bibitem{koehler2020nonlinear}
J.~K\"ohler, M.~A. M\"uller, and F.~Allg\"ower, ``A nonlinear tracking model
  predictive control scheme for dynamic target signals,'' \emph{Automatica},
  vol. 118, p. 109030, 2020.

\end{thebibliography}

\newpage
\section*{Appendix: Proof of Theorem~\ref{thm:continuity}}
\begin{proof}
Parts of the following proof are adapted from the proof of~\cite[Proposition 1]{berberich2021linearpart2_extended}.
Parts (i), (ii).a, and (ii).b are similar to~\cite[Proposition 1]{berberich2021linearpart2_extended}.
On the other hand, Part (ii).c addresses the issue that, in contrast to~\cite[Proposition 1]{berberich2021linearpart2_extended}, the nominal / model-based MPC problem~\eqref{eq:DD_MPC} does not contain a regularization of $\alpha(t)$ in the cost.\\
\textbf{(i). Proof of cost bound}\\
Since $J_L^*(\xi_t)\leq V(\xi_t)\leq\overline{V}$, Problem~\eqref{eq:DD_MPC} is feasible.
In the following, we use the optimal solution of Problem~\eqref{eq:DD_MPC} to define a candidate solution for Problem~\eqref{eq:DD_MPC_robust}.
To this end, let
\begin{align}\label{eq:thm_continuity_proof_uy_cand_def}
\hat{u}(t)=\bar{u}^*(t),\>\>\hat{y}(t)=\begin{bmatrix}\tilde{y}_{[t-n,t-1]}\\\bar{y}_{[0,L-1]}^*(t)\end{bmatrix}.
\end{align}
Further, define
\begin{align}\label{eq:Hux}
H_{\rmu\rmx}\coloneqq\begin{bmatrix}H_{L+n}(u^\rmd)\\H_1(x^\rmd_{[0,N-L]})
\end{bmatrix},
\end{align}
where $\{x^\rmd_k\}_{k=0}^{N-1}$ is the state trajectory corresponding to $(u^\rmd,y^\rmd)$.
Using Assumption~\ref{ass:pe} and~\cite[Corollary 2]{willems2005note}, $H_{\rmu\rmx}$ has full row rank.
We now choose
\begin{align}\label{eq:thm_continuity_proof_alpha_cand_def}
\hat{\alpha}(t)=H_{\rmu\rmx}^\dagger\begin{bmatrix}\hat{u}(t)\\x_{t-n}\end{bmatrix},
\end{align}
where $H_{\rmu\rmx}^\dagger$ is the Moore-Penrose inverse of $H_{\rmu\rmx}$.
Finally, the slack variable is chosen as
\begin{align}\label{eq:thm_continuity_proof_slack_def}
\hat{\sigma}(t)=&H_{L+n}(\tilde{y}^\rmd)\hat{\alpha}(t)-\hat{y}(t)\\\nonumber
\stackrel{\eqref{eq:thm_continuity_proof_uy_cand_def}-\eqref{eq:thm_continuity_proof_alpha_cand_def}}{=}&H_{L+n}(\varepsilon^\rmd)\hat{\alpha}(t)-\begin{bmatrix}\varepsilon_{[t-n,t-1]}\\0\end{bmatrix}.
\end{align}
This implies
\begin{align}\label{eq:thm_continuity_proof_sigma_bound}
\lVert \hat{\sigma}(t)\rVert_2^2\leq C_1\bar{\varepsilon}^2\lVert\hat{\alpha}(t)\rVert_2^2+C_2\bar{\varepsilon}^2
\end{align}
for some $C_1,C_2>0$.
Exploiting that the above candidate is feasible for Problem~\eqref{eq:DD_MPC_robust}, we infer
\begin{align}\nonumber
&\hat{J}_L^*(\tilde{\xi}_t)-J_L^*(\xi_t)\leq \lambda_{\alpha}\bar{\varepsilon}^{\beta_{\alpha}}\lVert\hat{\alpha}(t)\rVert_2^2
+\frac{\lambda_{\sigma}}{\bar{\varepsilon}^{\beta_{\sigma}}}\lVert\hat{\sigma}(t)\rVert_2^2\\\nonumber
\stackrel{\eqref{eq:thm_continuity_proof_alpha_cand_def},\eqref{eq:thm_continuity_proof_sigma_bound}}{\leq}
&(\lambda_{\alpha}\bar{\varepsilon}^{\beta_{\alpha}}+\lambda_{\sigma}C_1\bar{\varepsilon}^{2-\beta_{\sigma}})\lVert H_{\rmx\rmu}^\dagger\rVert_2^2
(\lVert\hat{u}(t)\rVert_2^2+\lVert x_{t-n}\rVert_2^2)\\\label{eq:thm_continuity_proof_cost_bound1}
&+\lambda_{\sigma}C_2\bar{\varepsilon}^{2-\beta_{\sigma}}.
\end{align}
To bound the first term on the right-hand side, we use that, by~\cite[Equation (16)]{berberich2021guarantees}, there exists $C_3>0$ satisfying
\begin{align}\label{eq:thm_continuity_proof_cost_bound2}
\lVert x_{t-n}\rVert_2^2\leq C_3\lVert\xi_t\rVert_2^2.
\end{align}
Further, by $\lVert\hat{u}(t)\rVert_2^2=\lVert\bar{u}^*(t)\rVert_2^2\leq\frac{1}{\lambda_{\min}(R)}V(\xi_t)$ and $V(\xi)\geq\frac{\lambda_{\min}(P)}{c_{\mathrm{IOSS}}}\lVert\xi\rVert_2^2$, we infer
\begin{align}\label{eq:thm_continuity_proof_cost_bound3}
\lVert\hat{u}(t)\rVert_2^2+\lVert\xi_t\rVert_2^2\leq C_4 V(\xi_t)
\end{align}
for some $C_4>0$.
Combining~\eqref{eq:thm_continuity_proof_cost_bound1}--\eqref{eq:thm_continuity_proof_cost_bound3}, we obtain
\begin{align}\nonumber
\hat{J}_L^*(\tilde{\xi}_t)\leq &J_L^*(\xi_t)+(C_5\bar{\varepsilon}^{\beta_{\alpha}}+C_6\bar{\varepsilon}^{2-\beta_{\sigma}})V(\xi_t)+C_7\bar{\varepsilon}^{2-\beta_{\sigma}}\\\label{eq:thm_continuity_proof_cost_bound_final}
\leq& J_L^*(\xi_t)+\beta_1(\bar{\varepsilon})
\end{align}
for some $C_i>0$, $i\in\bbi_{[5,7]}$, where $\beta_1\in\calK_{\infty}$ due to $V(\xi_t)\leq\overline{V}$ and $2-\beta_{\sigma}>0$ due to Assumption~\ref{ass:continuity}.\newpage
\textbf{(ii). Proof of~\eqref{eq:thm_continuity}}\\
\textbf{(ii).a Bound on $\lVert\hat{u}^*(t)-\tilde{u}(t)\rVert_2$}\\
\begin{subequations}\label{eq:DD_MPC_aux}
Consider now the (auxiliary) optimization problem
\begin{align}\label{eq:DD_MPC_aux_cost}
\underset{\alpha(t),\bar{u}(t),\bar{y}(t)}{\min}&\sum_{k=0}^{L-1}
\lVert\bar{u}_k(t)\rVert_R^2+\lVert\bar{y}_k(t)\rVert_Q^2+\lambda_{\alpha}\bar{\varepsilon}^{\beta_{\alpha}}\lVert\alpha(t)\rVert_2^2\\
\label{eq:DD_MPC_aux_hankel} \text{s.t.}\>\> &\>\begin{bmatrix}
\bar{u}(t)\\\bar{y}(t)+\tilde{\sigma}_1\end{bmatrix}=\begin{bmatrix}H_{L+n}(u^\rmd)\\H_{L+n}(y^\rmd)\end{bmatrix}\alpha(t),\\\label{eq:DD_MPC_aux_init}
&\>\begin{bmatrix}\bar{u}_{[-n,-1]}(t)\\\bar{y}_{[-n,-1]}(t)\end{bmatrix}=\begin{bmatrix}u_{[t-n,t-1]}\\y_{[t-n,t-1]}+\tilde{\sigma}_2\end{bmatrix},\\\label{eq:DD_MPC_aux_constraints}
&\>\bar{u}_k(t)\in\mathbb{U},\>k\in\mathbb{I}_{[0,L]},
\\\label{eq:DD_MPC_aux_TEC}
&\>\bar{u}_k(t)=0,\>\bar{y}_k(t)=0,\>k\in\bbi_{[L-n,L-1]},
\end{align}
\end{subequations}
where
\begin{align}\label{eq:thm_continuity_proof_sigma_tilde_def}
\tilde{\sigma}=\begin{bmatrix}\tilde{\sigma}_1\\\tilde{\sigma}_2\end{bmatrix}\coloneqq\begin{bmatrix}\hat{\sigma}^*(t)-H_{L+n}(\varepsilon^\rmd)\hat{\alpha}^*(t)\\
\varepsilon_{[t-n,t-1]}
\end{bmatrix}.
\end{align}
We denote the optimal solution of Problem~\eqref{eq:DD_MPC_aux} by $\tilde{\alpha}(t)$, $\tilde{u}(t)$, $\tilde{y}(t)$, and the optimal cost by $\tilde{J}_L$.
Since $J_L^*(\xi_t)\leq V(\xi_t)\leq\overline{V}$, Problem~\eqref{eq:DD_MPC} is feasible and thus, by Part (i) of the proof, Problem~\eqref{eq:DD_MPC_robust} is feasible as well.
The optimal solution of Problem~\eqref{eq:DD_MPC_robust} is feasible for Problem~\eqref{eq:DD_MPC_aux}, i.e.,
\begin{align}\label{eq:thm_continuity_proof_cost_bound5}
\tilde{J}_L\leq\hat{J}_L^*(\tilde{\xi}_t)-\frac{\lambda_{\sigma}}{\bar{\varepsilon}^{\beta_{\sigma}}}\lVert\hat{\sigma}^*(t)\rVert_2^2,
\end{align}
where the term involving $\hat{\sigma}^*(t)$ is due to the fact that Problem~\eqref{eq:DD_MPC_aux} does not contain a slack variable in the cost~\eqref{eq:DD_MPC_aux_cost}.
On the other hand, a feasible solution for Problem~\eqref{eq:DD_MPC_robust} can be defined via $\hat{u}(t)=\tilde{u}(t)$, $\hat{y}(t)=\tilde{y}(t)$, $\hat{\alpha}(t)=\tilde{\alpha}(t)$, and 
\begin{align}\label{eq:thm_continuity_proof_cost_bound9}
\hat{\sigma}(t)=\hat{\sigma}^*(t)+H_{L+n}(\varepsilon^\rmd)(\hat{\alpha}(t)-\hat{\alpha}^*(t)),
\end{align}
where we denote the corresponding cost by $\hat{J}_L'$.
By optimality, we have $\hat{J}_L^*(\tilde{\xi}_t)\leq\hat{J}_L'$.
Moreover, by definition, it holds that
\begin{align}\label{eq:thm_continuity_proof_cost_bound8}
\hat{J}_L'-\tilde{J}_L=\frac{\lambda_{\sigma}}{\bar{\varepsilon}^{\beta_{\sigma}}}\lVert\hat{\sigma}(t)\rVert_2^2.
\end{align}
From optimality, it follows that
\begin{align}\label{eq:thm_continuity_proof_cost_bound10}
\lVert\hat{\alpha}^*(t)\rVert_2^2\leq&\frac{\hat{J}_L^*(\tilde{\xi}_t)}{\lambda_{\alpha}\bar{\varepsilon}^{\beta_{\alpha}}},\\\label{eq:thm_continuity_proof_cost_bound11}
\lVert\hat{\sigma}^*(t)\rVert_2^2\leq&\bar{\varepsilon}^{\beta_{\sigma}}\frac{\hat{J}_L^*(\tilde{\xi}_t)}{\lambda_{\sigma}}.
\end{align}
Further, we have
\begin{align}\label{eq:thm_continuity_proof_cost_bound12}
\lVert\hat{\alpha}(t)\rVert_2^2=\lVert\tilde{\alpha}(t)\rVert_2^2\leq\frac{\tilde{J}_L}{\lambda_{\alpha}\bar{\varepsilon}^{\beta_{\alpha}}}
\stackrel{\eqref{eq:thm_continuity_proof_cost_bound5}}{\leq}
\frac{\hat{J}_L^*(\tilde{\xi}_t)}{\lambda_{\alpha}\bar{\varepsilon}^{\beta_{\alpha}}}.
\end{align}
Problem~\eqref{eq:DD_MPC_robust} is strongly convex in $\hat{u}$, i.e., there exists $c_{\rmu,\rmc}>0$ such that
\begin{align}\label{eq:thm_continuity_proof_cost_bound6}
\lVert\hat{u}(t)-\hat{u}^*(t)\rVert_2^2\leq c_{\rmu,\rmc}(\hat{J}_L'-\hat{J}_L^*(\tilde{\xi}_t)),
\end{align}
compare~\cite[Inequality (11)]{koehler2020nonlinear}.
Together with $\hat{u}(t)=\tilde{u}(t)$, we infer
\begin{align*}
\lVert\hat{u}^*(t)-\tilde{u}(t)\rVert_2^2
\stackrel{\eqref{eq:thm_continuity_proof_cost_bound6}}{\leq}
&c_{\rmu,\rmc}(\hat{J}_L'-\hat{J}_L^*(\tilde{\xi}_t))\\
\stackrel{\eqref{eq:thm_continuity_proof_cost_bound5}}{\leq}
&c_{\rmu,\rmc}\left(\hat{J}_L'-\tilde{J}_L-\frac{\lambda_{\sigma}}{\bar{\varepsilon}^{\beta_{\sigma}}}\lVert\hat{\sigma}^*(t)\rVert_2^2\right)\\
\stackrel{\eqref{eq:thm_continuity_proof_cost_bound8}}{=}
&c_{\rmu,\rmc}\frac{\lambda_{\sigma}}{\bar{\varepsilon}^{\beta_{\sigma}}}(\lVert\hat{\sigma}(t)\rVert_2^2-\lVert\hat{\sigma}^*(t)\rVert_2^2).
\end{align*}
Using~\eqref{eq:thm_continuity_proof_cost_bound9} as well as
\begin{align*}
\lVert a\rVert_2^2-\lVert b\rVert_2^2\leq \lVert a-b\rVert_2^2+2\lVert a-b\rVert_2\lVert b\rVert_2,
\end{align*}
we infer
\begin{align*}
&\lVert\hat{\sigma}(t)\rVert_2^2-\lVert\hat{\sigma}^*(t)\rVert_2^2\\
\leq&\lVert H_{L+n}(\varepsilon^\rmd)(\hat{\alpha}(t)-\hat{\alpha}^*(t))\rVert_2^2\\
&+2\lVert H_{L+n}(\varepsilon^\rmd)(\hat{\alpha}(t)-\hat{\alpha}^*(t))\rVert_2
\lVert\hat{\sigma}^*(t)\rVert_2\\
\stackrel{\eqref{eq:thm_continuity_proof_cost_bound10}-\eqref{eq:thm_continuity_proof_cost_bound12}}{\leq}
&C_8\bar{\varepsilon}^{2-\beta_{\alpha}}\hat{J}_L^*(\tilde{\xi}_t)
+C_9\bar{\varepsilon}^{1+\frac{\beta_{\sigma}-\beta_{\alpha}}{2}}\hat{J}_L^*(\tilde{\xi}_t)
\end{align*}
for suitably defined $C_8,C_9>0$.
Thus, using~\eqref{eq:thm_continuity_proof_cost_bound_final} as well as $V(\xi_t)\leq\overline{V}$, we have
\begin{align}\label{eq:thm_continuity_proof_input_bound1}
\lVert\hat{u}^*(t)-\tilde{u}(t)\rVert_2\leq\beta_2(\bar{\varepsilon}),
\end{align}
where
\begin{align*}
&\beta_2(\bar{\varepsilon})\coloneqq\\
&\sqrt{c_{\rmu,\rmc}\lambda_{\sigma}\left(C_8\bar{\varepsilon}^{2-\beta_{\alpha}-\beta_{\sigma}}+C_9\bar{\varepsilon}^{\frac{1}{2}(2-\beta_{\alpha}-\beta_{\sigma})}\right)(\overline{V}+\beta_1(\bar{\varepsilon}))}.
\end{align*}
Note that $\beta_2\in\calK_{\infty}$ due to $\beta_{\alpha}+\beta_{\sigma}<2$.\\
\textbf{(ii).b Bound on $\lVert\tilde{u}(t)-u'(t)\rVert_2$}\\
It remains to derive a bound on $\lVert \tilde{u}(t)-\bar{u}^*(t)\rVert_2$, which, together with~\eqref{eq:thm_continuity_proof_input_bound1}, will imply~\eqref{eq:thm_continuity}.
Feasibility of Problem~\eqref{eq:DD_MPC} implies feasibility of~\eqref{eq:DD_MPC_aux} with $\tilde{\sigma}=0$.
We denote the optimal input of Problem~\eqref{eq:DD_MPC_aux} with $\tilde{\sigma}=0$ by $u'(t)$.
Similar to~\cite[Proposition 1]{berberich2021linearpart2_extended}, we can use standard arguments from multi-parametric quadratic programming~\cite{bemporad2002explicit} to infer
\begin{align}\label{eq:thm_continuity_proof_u_prime_tilde_bound}
&\lVert \tilde{u}(t)-u'(t)\rVert_2\leq C_{10}\lVert\tilde{\sigma}\rVert_2\\\nonumber
\stackrel{\eqref{eq:thm_continuity_proof_sigma_tilde_def},\eqref{eq:thm_continuity_proof_cost_bound10},\eqref{eq:thm_continuity_proof_cost_bound11}}{\leq}
&C_{11}\left(\bar{\varepsilon}+\bar{\varepsilon}^{\frac{\beta_{\sigma}}{2}}+\bar{\varepsilon}^{1-\frac{\beta_{\alpha}}{2}}\right)
\end{align}
for some $C_{10},C_{11}>0$, where the last inequality also uses
\begin{align*}
\hat{J}_L^*(\tilde{\xi}_t)\stackrel{\eqref{eq:thm_continuity_proof_cost_bound_final}}{\leq}
J_L^*(\xi_t)+\beta_1(\bar{\varepsilon})\leq C_{12}
\end{align*}
for some $C_{12}>0$ due to $J_L^*(\xi_t)\leq\overline{V}$ and bounded $\bar{\varepsilon}$.\\
\textbf{(ii).c Bound on $\lVert u'(t)-\bar{u}^*(t)\rVert_2$}\\
\begin{subequations}\label{eq:DD_MPC_proof2}
Finally, we derive a bound on $\lVert u'(t)-\bar{u}^*(t)\rVert_2$.
Recall that $\bar{u}^*(t)$ is the optimal input of Problem~\eqref{eq:DD_MPC}, whereas $u'(t)$ is the optimal input of Problem~\eqref{eq:DD_MPC} when adding the term $\lambda_{\alpha}\bar{\varepsilon}^{\beta_{\alpha}}\lVert\alpha(t)\rVert_2^2$ to the cost.
Since, by assumption, $\bbu$ is a polytope and $\bby=\bbr^p$, Problem~\eqref{eq:DD_MPC} (including the cost term $\lambda_{\alpha}\bar{\varepsilon}^{\beta_{\alpha}}\lVert\alpha(t)\rVert_2^2$) can be transformed to the following generic quadratic program
\begin{align}\label{eq:DD_MPC_proof2_cost}
\underset{z,\alpha}{\min}&\lVert z\rVert_2^2+\bar{\varepsilon}^{\beta_{\alpha}}\lVert\alpha\rVert_2^2\\
\label{eq:DD_MPC_proof2_hankel} \text{s.t.}\>\> &\>
H\alpha=z,\\\label{eq:DD_MPC_proof2_eq}
&\>A_{\mathrm{eq}}z=b_{\mathrm{eq}},\>A_{\mathrm{ineq}}z\leq b_{\mathrm{ineq}}.
\end{align}
\end{subequations}
\begin{subequations}\label{eq:DD_MPC_proof3}
We write $z^*(\bar{\varepsilon})$ for the optimal value $z$ of Problem~\eqref{eq:DD_MPC_proof2} depending on $\bar{\varepsilon}$.
In the following, we show that $\lVert z^*(\bar{\varepsilon})-z^*(0)\rVert_2$ can be bounded in terms of the noise level, which directly implies an analogous bound on $\lVert u'(t)-\bar{u}^*(t)\rVert_2$.
Consider the following quadratic program
\begin{align}\label{eq:DD_MPC_proof3_cost}
\underset{z,\alpha,w}{\min}&\lVert z\rVert_2^2+\bar{\varepsilon}^{\beta_{\alpha}}\lVert\alpha\rVert_2^2\\
\label{eq:DD_MPC_proof3_hankel} \text{s.t.}\>\> &
\alpha=H^\dagger z+(I-H^\dagger H)w,\\\label{eq:DD_MPC_proof3_eq}
&\>A_{\mathrm{eq}}z=b_{\mathrm{eq}},\>A_{\mathrm{ineq}}z\leq b_{\mathrm{ineq}},
\end{align}
\end{subequations}
\begin{subequations}\label{eq:DD_MPC_proof4}
where $H^\dagger$ denotes the Moore-Penrose inverse of $H$.
Problem~\eqref{eq:DD_MPC_proof3} is equivalent to Problem~\eqref{eq:DD_MPC_proof2} since~\eqref{eq:DD_MPC_proof3_hankel} parametrizes the solution space of~\eqref{eq:DD_MPC_proof2_hankel}.
We can eliminate this constraint to arrive at
\begin{align}\label{eq:DD_MPC_proof4_cost}
\underset{z,w}{\min}&\lVert z\rVert_2^2+\bar{\varepsilon}^{\beta_{\alpha}}\lVert H^\dagger z+(I-H^\dagger H)w\rVert_2^2\\\label{eq:DD_MPC_proof4_eq}
\text{s.t.}\>\>&A_{\mathrm{eq}}z=b_{\mathrm{eq}},\>A_{\mathrm{ineq}}z\leq b_{\mathrm{ineq}}.
\end{align}
\end{subequations}
By orthogonality, we infer
\begin{align*}
\lVert H^\dagger z+(I-H^\dagger H)w\rVert_2^2
=\lVert H^\dagger z\rVert_2^2+\lVert(I-H^\dagger H)w\rVert_2^2.
\end{align*}
\begin{subequations}\label{eq:DD_MPC_proof5}
Thus, the optimal solution of Problem~\eqref{eq:DD_MPC_proof4} satisfies $(I-H^\dagger H)w=0$ and, therefore, Problem~\eqref{eq:DD_MPC_proof4} is equivalent to
\begin{align}\label{eq:DD_MPC_proof5_cost}
\underset{z}{\min}&\>z^\top\left(I+\bar{\varepsilon}^{\beta_{\alpha}}(H^{\dagger})^\top H^{\dagger}\right)z\\\label{eq:DD_MPC_proof5_eq}
\text{s.t.}\>\>&A_{\mathrm{eq}}z=b_{\mathrm{eq}},\>A_{\mathrm{ineq}}z\leq b_{\mathrm{ineq}}.
\end{align}
\end{subequations}
Denoting the optimal cost of Problem~\eqref{eq:DD_MPC_proof5} for a given $\bar{\varepsilon}>0$ by $S(\bar{\varepsilon})$, it is not hard to see that
\begin{align}\label{eq:DD_MPC_proof43}
S(\bar{\varepsilon})&\leq (1+C_{13}\bar{\varepsilon}^{\beta_{\alpha}})S(0),\\
S(0)&\leq S(\bar{\varepsilon})
\end{align}
for some $C_{13}>0$ and any $\bar{\varepsilon}>0$.
Using strong convexity of Problem~\eqref{eq:DD_MPC_proof5}, there exists $C_{14}>0$ such that
\begin{align*}
\lVert z^*(\bar{\varepsilon})-z^*(0)\rVert_2^2\leq &C_{14}
(S(\bar{\varepsilon})-S(0))\\
\stackrel{\eqref{eq:DD_MPC_proof43}}{\leq} &C_{14}C_{13}\bar{\varepsilon}^{\beta_{\alpha}}S(0),
\end{align*}
compare~\cite[Inequality (11)]{koehler2020nonlinear}.
Applied to Problem~\eqref{eq:DD_MPC} with $S(0)=J_L^*(\xi_t)\leq V(\xi_t)\leq\overline{V}$, there exists $C_{15}>0$ such that
\begin{align}\label{eq:DD_MPC_proof27}
\lVert u'(t)-\bar{u}^*(t)\rVert_2\leq C_{15}\bar{\varepsilon}^{\frac{\beta_{\alpha}}{2}}.
\end{align}
Combining~\eqref{eq:thm_continuity_proof_input_bound1},~\eqref{eq:thm_continuity_proof_u_prime_tilde_bound}, and~\eqref{eq:DD_MPC_proof27}, there exists $\beta_{\rmu}\in\calK_{\infty}$ satisfying~\eqref{eq:thm_continuity} which concludes the proof.
\end{proof}

\end{document}